\documentclass[onecolumn,floatfix,aps,nofootinbib]{revtex4}
\usepackage[dvips]{epsfig}
\usepackage[english]{babel}
\usepackage{bbm}
\usepackage{verbatim}
\usepackage{array}
\usepackage{bm} 
\usepackage{amsmath}
\usepackage{amsfonts}
\usepackage{amssymb}
\usepackage{hyperref}
\usepackage{cases}
\usepackage[dvipsnames]{xcolor}
\hypersetup{colorlinks=true,linkbordercolor=Blue,linkcolor=Blue, citecolor=Blue}
\usepackage{indentfirst} 
\usepackage[titletoc]{appendix}
\usepackage{subeqnarray}
\usepackage{setspace}
\usepackage{indentfirst} 

\usepackage{gensymb}
\usepackage{xcolor}
\usepackage{calrsfs}

\usepackage{wasysym}
\usepackage[all,cmtip]{xy}

    \usepackage{amsmath,amsfonts,amssymb,amsthm,mathrsfs,bbm,braket}

    \newcommand{\dd}{\widetilde{d}}

    \newtheorem{defi}{Definition}[section]
    \newtheorem{prop}{Proposition}[section]
    
    \newtheorem{lemma}{Lemma}[section]

    \numberwithin{equation}{section}

\usepackage{bbold}

\begin{document}

\title{A geometrical approach to nontrivial topology via exotic spinors}

\author{J. M. Hoff da Silva} 
\email{julio.hoff@unesp.br}
\affiliation{Departamento de F\'isica, Universidade
Estadual Paulista, UNESP, Av. Dr. Ariberto Pereira da Cunha, 333, Guaratinguet\'a, SP,
Brazil.}

\author{R. T. Cavalcanti} 
\email{rogerio.cavalcanti@ufabc.edu.br}
\affiliation{Center of Mathematics, Federal University of ABC, 09210-580, Santo Andr\'e, Brazil.\\ Departamento de F\'isica, Universidade Estadual Paulista, UNESP, Av. Dr. Ariberto Pereira da Cunha, 333, Guaratinguet\'a, SP, Brazil.}

\author{D. Beghetto} 
\email{dino.junior@ifnmg.edu.br}
\affiliation{Instituto Federal do Norte de Minas Gerais.
	Rodovia MG 202, km 392, Sub Trecho: Arinos / Entroncamento de Uruana de Minas, Arinos, MG. Brazil.}

\author{G. M. Caires da Rocha} 
\email{gabriel.marcondes@unesp.br}
\affiliation{Departamento de F\'isica, Universidade
	Estadual Paulista, UNESP, Av. Dr. Ariberto Pereira da Cunha, 333, Guaratinguet\'a, SP,
	Brazil. }



\begin{abstract}
Exotic spinors arise in non-simply connected base manifolds due to the nonequivalent spinor structure. The dynamics of exotic spinors are endowed with an additional differential factor. In this work, we merge the exotic spinor scenario with Cartan's spinor viewpoint, according to which a given spacetime point is understood as a kind of composition of spinor entries. As a result, we arrive at a geometrical setup in which the Minkowski metric is perturbed by elements reflecting the nontrivial topology. Such corrections shall be felt by any physical system studied with the resulting bilinear form. Within the flat spacetime context, we investigate quasinormal modes arising from the interference of nontrivial topology in the scalar field dispersion relation. 
\end{abstract}		
\maketitle
\section{Introduction}

Among the plethora of different descriptions of spinor fields, there is a pretty distinct approach due to \'Elie Cartan \cite{cart}. According to Cartan, spinors are projections of the Riemann sphere (obtained by slicing the light cone at a fixed time) onto the complex plane. This program was also expanded to accommodate all surfaces in Minkowski space \cite{penro}. As a result, every spacetime point could be described by a somewhat spinor components combination. In fact, if $\xi$ and $\eta$ are the spinorial entries, then $(t,x,y,z) = \big(t(\xi,\eta), x(\xi,\eta), y(\xi,\eta), z(\xi,\eta)\big)$ and the spinor entries act as elements of a pre-geometry.  

On the other hand, if the base manifold accommodating physical fields is not simply connected and still the general conditions for the existence of spinors are maintained, then there is more than one spinorial structure in order. Hence the arising of exotic spinors is possible \cite{3,exot}. To a large extent, exotic and usual spinors are indistinguishable, but the exotic spinor dynamics also brings information on the nontrivial topology. Thus, in the case of a multi-connected base manifold, an exotic spinor $\tilde{\psi}$ has dynamics dictated by $(i(d+id\theta)-m)\tilde{\psi}=0$ where $\theta$ is a real scalar function encoding the nontrivial topological information. Exotic spinors effects were investigated in superconductivity \cite{petry} and field theory \cite{field,f2}, in association with mass dimension one fermionic field \cite{nos} and minimal length fermionic systems \cite{so}. Its impacts on heat kernel coefficients were studied in \cite{ozer}. It turns out, however, that the correction of the exotic spinor dynamics takes place in spinor entries. A natural question in trying to merge these two viewpoints is: what spacetime geometry would result from a pre-geometry performed by exotic spinor entries? In other words, if an arbitrary spacetime point is related to (exotic) $\xi$'s and $\eta$'s, whose derivative operator is the one encompassing the correction due to the nontrivial topology, what impact does it have in the spacetime geometry? While these questions are exposed here from a motivational theoretical perspective, it is relevant to note that a signature for a global multiply connected Universe was evinced by the variance of the cosmic microwave background temperature gradient \cite{MC}. 

To approach the previous questions, we start implementing the exotic correction into spacetime differentials and look at the outcome due to the nontrivial topology. We then move to the study of bilinear forms and the accordingly modifications. Although the scenario demands thoughtful treatment, changes in the bilinear form are made reasonably manageable. Thus, with due care, several consequences are evinced. The modified metric carries terms reflecting the nontrivial topology and explicitly depends on the spacetime point; that is what we call `geometrization of topology'. The usual Minkowski symmetries could be seen as approximate low-energy spacetime symmetries. This framework may be, argumentatively, a source for Lorentz violating models \cite{VL}. Also, we argue on a set of approximations on the term reflecting the nontrivial topology such that it is still possible to work in a flat spacetime. These approximations are referred to a given scale, discussed throughout the text, in which the effects of nontrivial topology are reasonably small. However, notwithstanding, signatures of interesting physical consequences do appear. One of the welcome features this approach brings is the possibility of investigating the physical effects of nontrivial topology for fermionic and bosonic fields. We thus study the Klein-Gordon-like equation and compute the resulting dispersion relation. It is shown that the field modes are indeed affected by the nontrivial topology, here performed by a local 'object' causative of non simply connectivity in a given region, e.g., a (higher dimensional) localized hole obstructing closed loops to be contracted back into a dot. Consequently, it leads to an expansion in quasinormal-like modes, usually (but not exclusively) related to the interaction of field modes with the background of a black hole \cite{quasi}. 

This work is organized as follows: Section II presents the mathematical base underlining the main motivations. Section III comprises our main results. It starts implementing the idea and analyzing the primary consequences, also discussing the domain of validity of our procedures; the resulting bilinear form is then investigated, and the study of quasinormal-like behavior is performed via a Fourier transform of the scalar field dynamical equation. In Section IV, we conclude. Finally, in the Appendix, we discuss relevant aspects concerning differential forms and exterior derivatives.

\section{a short review on the mathematical formalism}

This section briefly reviews two relevant views of spinor formalism, both of particular interest to this paper. First, we start recalling the essential steps leading to spacetime points written via spinor entries, after which we move to the exotic spinors. 

\subsection{Spacetime spinorial structure}

We shall explore here a viewpoint whose roots are found in the work of Cartan \cite{cart}. Consider the background given by the Minkowski space, endowed with its usual basis and $v = (t,\,x,\,y,\,z)$ be a light-like vector so that
\begin{equation}
t^2
-x^2
-y^2
-z^2
= 0.
\end{equation}
In this scenario, one can obtain Minkowski's coordinates in terms of spinor's coordinates through a typical geometrical construction. By intercepting the light-cone with a hyper-plane defined by $t=1$. The so-called celestial sphere is now described by $x^2+y^2+z^2=1$. Such a sphere can be brought to a stereographic projection on a complex plane intersecting the sphere on $z=0$. That projection can be constructed by looking at lines that start on the sphere's north pole and reach the complex plane, starting from the north pole and passing through a point $P$ of the sphere. The sphere's coordinates $P=(1,x,\,y,\,z)$ will be taken to complex number $\beta$ given by $\beta=(x+iy)/(1-z)$. To avoid the singularity at the north pole, one can define another set of coordinates by the equation $\beta=\zeta/\chi$, with $\zeta$ and $\chi$ being complex numbers. The pair $(\zeta,\chi)$ equal to $(1,0)$ corresponds to the north pole, regular at infinity. That pair defines what is called a spinor, and Minkowski's coordinates can be gotten \`a la Cartan, written in terms of spinor coordinates,
\begin{equation}\label{}
t=1,\;
x=\frac{\zeta\overline{\chi}+\chi\overline{\zeta}}{\zeta\overline{\zeta}+\chi\overline{\chi}},\;
y=\frac{\zeta\overline{\chi}-\chi\overline{\zeta}}{i(\zeta\overline{\zeta}+\chi\overline{\chi})},\;
z=\frac{\zeta\overline{\zeta}-\chi\overline{\chi}}{\zeta\overline{\zeta}+\chi\overline{\chi}}.\;
\end{equation}
The results presented so far are such that $t=1$. To get a description of a point with any other $t$, one just need to multiply $P$'s coordinates \cite{pm} by $(\zeta\overline{\zeta}+\chi\overline{\chi})/\sqrt{2}$, leading to
\begin{equation}\label{}
t=\frac{\zeta\overline{\zeta}+\chi\overline{\chi}}{\sqrt{2}},\;
x=\frac{\zeta\overline{\chi}+\chi\overline{\zeta}}{\sqrt{2}},\;
y=\frac{\zeta\overline{\chi}-\chi\overline{\zeta}}{i\sqrt{2}},\;
z=\frac{\zeta\overline{\zeta}-\chi\overline{\chi}}{\sqrt{2}}.\;
\end{equation}
The complex numbers $\zeta$ and $\chi$ may be straightforwardly found as  
\begin{equation}
\zeta =\pm\bigg(\tfrac{\sqrt{2}}{2}(t+z)\bigg)^{\!1/2},\;
\chi=\pm\bigg(\tfrac{\sqrt{2}}{2}(t-z)\bigg)^{\!1/2}.\;
\end{equation}\\
Just as in the Euclidean case studied by Cartan \cite{cart}, it is impossible to fix consistent signs for all light-like vectors to obtain a solution that varies continuously with respect to all those light-like vectors. Furthermore, if the vector $v$ is rotated by an angle $\alpha$, $v\mapsto e^{i\alpha}v$, thus the pair $(\zeta,\,\chi) $ will be transformed via $(\zeta,\,\chi)\mapsto(e^{i\tfrac{\alpha}{2}}\zeta,\,e^{i\tfrac{\alpha}{2} }\chi)$. In particular, if $\alpha=2\pi$, the vector $v$ returns to its original version, and $(\zeta,\,\chi)$ receives a negative sign.

Notice that the following product 
\begin{equation}\label{Def_V}
\begin{pmatrix}
\zeta\\
\chi
\end{pmatrix}
\begin{pmatrix}
\bar{\zeta}&\bar{\chi}
\end{pmatrix}
=
\begin{pmatrix}
\zeta\bar{\zeta}&\zeta\bar{\chi}\\
\chi\bar{\zeta}&\chi\bar{\chi}
\end{pmatrix}
=
\frac{1}{\sqrt{2}}
\begin{pmatrix}
t+z&x+iy\\
x-iy&t-z
\end{pmatrix}
\equiv
V,
\end{equation}
is such that $\det(V)
=
t^2
-x^2
-y^2
-z^2$. Hence, it is possible to introduce complex $2\times 2$ matrices, say $\lambda$, which perform Lorentz transformations through the conjugation
\begin{equation}\label{eq_trans_V}
V\mapsto V'
=
\lambda
V
\lambda^{\dagger},
\end{equation}
where the metric invariance (the invariance of $\det(V)$) is attained requiring $|\det(\lambda)|=1$. Therefore, the $\lambda$ matrices will be elements of the $SL(2,\mathbb{C})$ group. Then, it can be seen that if we replace (\ref{Def_V}) in (\ref{eq_trans_V}), we have
\begin{equation}
\lambda
V
\lambda^{\dagger}
=
\lambda
\begin{pmatrix}
\zeta\\
\chi
\end{pmatrix}
\begin{pmatrix}
\bar{\zeta}&\bar{\chi}
\end{pmatrix}
\lambda^{\dagger},
\end{equation}
from which $\lambda$ is regarded as a spinorial transformation, so that $\lambda \in Spin(1,3)$ and $Spin(1,3)\cong SL(2,\mathbb{C})$ . In addition, as $-\lambda$ and $\lambda$ perform the same effect in preserving the metric, a classical result states that $SL(2,\mathbb{C} )/\mathbb{Z}_2 \cong SO(1,3)$.
 
\subsection{Exotic Spinor structures}

What follows is a brief discussion concerning the construction of exotic spinors. We shall do this by comparing it with the construction of non-exotic spinors. One must notice that the spacetime topology will be a major condition for the exotic spinors to exist. First, let us define a spin structure on a 4-dimensional Lorentzian spacetime manifold $M$. Let $P_{Spin(1,3)} \xrightarrow{\pi_s} M$ be an orthogonal frame bundle, and $s: P_{Spin(1,3)} \rightarrow P_{SO(1,3)}$ a double cover such that $\pi_s = \pi \circ s$ for $\pi : P_{SO(1,3)} \rightarrow M$. A spin structure is a pair formed by the principal bundle and the double cover, $\left(P_{Spin(1,3)}, s\right)$.

Regarding the existence of spin structures, one has to investigate the so-called Stiefel-Whitney classes $w_i(E) \in H^i(B,\mathbb{Z}_2)$, defined for a real vector bundle $E \xrightarrow{\phi} B$. The vector spaces $H^i(B,\mathbb{Z}_2)$ are referred to as the cohomology groups of $B$ with coefficients in $\mathbb{Z}_2$. Their elements (the Stiefel-Whitney classes) are called characteristic classes. They are related to invariants on vector bundles and generate all the ordinary cohomology classes with coefficients in $\mathbb{Z}_2$. In particular, the Stiefel-Whitney second class establishes the existence of spin structures on manifolds: a Riemannian manifold admits spin structure if, and only if, its Stiefel-Whitney second class is null \cite{mil, hatcher,3}. The set of spin structures on a manifold $M$ is labeled by the elements of the first cohomology group $H^1(M,\mathbb{Z}_2)$. When the manifold $M$ is non-simply connected, i.e., $H^1(M,\mathbb{Z}_2) \neq 0$, exotic spin structures $\left(\tilde{P}_{Spin(1,3)}, \tilde{s}\right)$ are allowed to exist. Such spin structures are inequivalent to the usual ones. Consequently, exotic spinors emerge as sections of the spinor bundle associated with the principal bundle $\tilde{P}_{Spin(1,3)}$.  Finally, one concludes that the topological aspects of the spacetime manifold lie behind the very birth of exotic spinors. We shall pursue this train of thought.

Let us start by defining the two spin structures $P:=(P_{Spin(1,3)}, s)$ and $\tilde{P}:=(\tilde{P}_{Spin(1,3)}, \tilde{s})$. Then, $P$ and $\tilde{P}$ are called equivalents if there exists a $Spin(1,3)-$equivariant mapping $q:P \rightarrow \tilde{P}$ so that the following diagram commutes:
\begin{equation}
\xymatrix{
	P \ar[rdd]_s \ar[rr]^q & & \tilde{P} \ar[ldd]^{\tilde{s}}\nonumber \\
	& & \\
	& P_{SO(1,3)}
}
\end{equation}
Consider the group homomorphism $\varsigma: Spin(1,3) \rightarrow SO(1,3)$ such that $\ker(\varsigma) \cong \mathbb{Z}_2$. Let $\cup_{i \in I} U_i$ be an open cover for $M$, with transition functions defined as
\begin{equation}
a_{ij}: U_i \cap U_j \rightarrow SO(1,3),
\end{equation}
such that $a_{ij} \circ a_{jk} = a_{ik}$ on $U_i \cap U_j \cap U_k$. For a spin structure $P$ on $M$, there is \cite{petry,nos} a system of transition functions
\begin{equation}
h_{ij}: U_i \cap U_j \rightarrow Spin(1,3),
\end{equation}
such that
\begin{equation}
\varsigma \circ h_{ij} = a_{ij};\;\; h_{ij} \circ h_{jk} = h_{ik}.
\end{equation}
Similarly, one has
\begin{equation}
\tilde{h}_{ij}: U_i \cap U_j \rightarrow \widetilde{Spin}(1,3) = Spin(1,3).
\end{equation}
What follows is that two spin structures $P$ and $\tilde{P}$ are respectively described by the maps $h_{ij}$ and $\tilde{h}_{ij}$, such that $\varsigma \circ h_{jk} = a_{jk} = \varsigma \circ \tilde{h}_{jk}$.

Now, define $\xi: M \rightarrow \mathbb{C}$ by
\begin{equation}
\xi(x) = \xi_i^2(x),
\end{equation}
where $\xi_i: U_i \rightarrow \mathbb{C}$ are specific unimodular functions satisfying $\xi_i(x) \in U(1)$ for each $x \in U_i \subset M$. These unimodular functions are called generators of the cocycles $c_{ij}$, which are maps defined by $h_{ij}(x) = \tilde{h}_{ij}(x) c_{ij}$, such that
\begin{equation}
c_{ij}: U_i \cap U_j \rightarrow \ker(\varsigma)=\mathbb{Z}_2 \hookrightarrow Spin(1,3),
\end{equation}
with $c_{ij} \circ c_{jk} = c_{ik}$ (the nontrivial elements of $H^1(M,\mathbb{Z}_2) \neq 0$, in fact). This construction defines a one-to-one correspondence between inequivalent spin structures and $H^1(M,\mathbb{Z}_2)$ \cite{3}. Considering a non-exotic spinor $\Psi \in \text{sec} P_{Spin(1,3)} \times \mathbb{C}^4$ and an exotic spinor $\tilde{\Psi} \in \text{sec} \tilde{P}_{Spin(1,3)} \times \mathbb{C}^4$, one can define a bundle mapping $f$ as
\begin{eqnarray}
f: & \tilde{P}_{Spin(1,3)} \times \mathbb{C}^4 & \rightarrow P_{Spin(1,3)} \times \mathbb{C}^4 \nonumber \\
& \tilde{\Psi}_i & \mapsto f(\tilde{\Psi}_i) = \Psi_i,
\end{eqnarray}
such that \cite{petry,SC,f2,nos}
\begin{equation}\label{derivadacovariante}
\tilde{\nabla}_X f(\tilde{\Psi}) = f(\nabla_X \tilde{\Psi}) + \frac{1}{2} \left( X \cdot (\xi^{-1} d\xi) \right) f(\tilde{\Psi})
\end{equation}
is applied\footnote{Generally, in (\ref{derivadacovariante}), left contraction is necessary for the second term. Here, the inner product will be used as it suffices for our purposes.} to all $\Psi \in \mathrm{sec}P_{Spin(1,3)} \times \mathbb{C}^4$ and all vector field $X \in M$. Therefore, even in the flat manifold case, a correction in the derivative is expected, and a replacement as   
\begin{equation}
\partial_\mu\mapsto \partial_\mu+\xi^{-1}(x)\partial_\mu \xi(x),\label{nvf1}
\end{equation} is in order. 

\section{consequences of the nontrivial topology}

This section aims to implement a topological correction in the spacetime geometry, motivated by merging the Cartan spinor view with exotic counterparts. As stated before, when the topology is nontrivial, the spinor dynamic is corrected, which means that the Dirac operator shall be changed. Nevertheless, the correction acts fundamentally in spinor entries.

Let us take a glance at the usual differentiation case for a function $f:\mathbb{R}^{n}\rightarrow \mathbb{R}$. Let $h=\sum\limits_{i}h^{i}e_{i}\in\mathbb{R}^{n}$ and $x_{0}\in U \subset \mathbb{R}^{n}$, where $U$ is a given open set. The differential of $f$ in $x_0$ acting in $h$ reads
\begin{equation}\label{df}
df(x_{0})(h)
=
df(x_{0})\bigg(\sum\limits_{i}h^{i}e_{i}\bigg)
=
\sum\limits_{i}df(x_{0})(e_{i})h^{i}
=
\sum\limits_{i}\frac{\partial f(x_{0})}{\partial x^{i}}h^{i},
\end{equation} as, naturally, $df(x_{0})(e_{i})=\partial f(x_{0})/\partial x^{i}$.  Employing the linear orthogonal projections $\pi^i$ 
\begin{equation}\label{proj}
\begin{split}
\pi^{i}:\mathbb{R}^{n}\to\mathbb{R}\qquad\qquad\qquad\quad\,\,\\
(x^1,\cdots,x^{i},\cdots,x^{n})\mapsto\pi^{i}(x^1,\cdots,x^{i},\cdots,x^{n})=x^{i},
\end{split}\!\!\!\!\!\!\!\!\!\!\!\!\!\!\!\!
\end{equation} one has $h^i=dx^i(h)=d\pi^i(h)$, so that for an arbitrary $h\in\mathbb{R}^n$ the final differential result is the quite familiar one 
\begin{equation}\label{dff}
df(x_{0})=\sum\limits_{i}\frac{\partial f}{\partial x^{i}}(x_{0})dx^{i}.
\end{equation}
 
The implementation of nontrivial topology goes as follows: we shall deal with the possibility of a multiply connected open set $U$. This concept was done more precisely in the last section. Here we shall expose the main ideas simplifying the presentation and eventually making it more formal as the necessity appears. By now, we only emphasize the notation $\widetilde{\mathbb{R}^n}$ for a space similar to $\mathbb{R}^n$ but encompassing at least one open set with nontrivial topology. As stated before, spacetime points may be faced as spinor entries condensation so that\footnote{We kept the Latin index for a while, but all the discussion here may be straightforwardly generalized to a pseudo-Euclidean spacetime $(n=p+q)$. The usual tensorial notation will be introduced during the presentation when discussing bilinear forms.} $x^i\sim (\zeta\bar{\zeta})^i$. In this vein, the orthogonal projections are such that
\begin{equation}\label{dpi}
d\pi^{i}=d(\overline{\zeta}\zeta)^{i}=\sum_j\frac{\partial}{\partial x_j}(\overline{\zeta}\zeta)^{i}dx^j,    
\end{equation} leading, by its turn, to 
\begin{equation}\label{dpin}
d\pi^{i}=\sum_j(\overline{\zeta}\partial_j\zeta+\partial_j\overline{\zeta}\zeta)^{i}dx^j.
\end{equation}

This factorization makes explicit the fact that the differentiation is taken over spinorial entries; therefore, we now implement a correction on the partial derivative operator motivated by the previous section's construction setting  $\partial_j\mapsto \partial_j+\partial_j \theta(x)$, with $\theta(x)\in\mathbb{R}$. By implementing it into (\ref{dpin}), the coefficient reads  
\begin{equation}\label{}
\bigg(\overline{\zeta}[(\partial_j+\partial_j\theta)\zeta]
+
[(\partial_j+\partial_j\theta)\overline{\zeta}]
\zeta
\bigg)^{i}
=
\bigg(
\overline{\zeta}\partial_j\zeta
+
\partial_j\theta\overline{\zeta}\zeta
+
\partial_j\overline{\zeta}\zeta
+
\partial_j\theta\overline{\zeta}\zeta
\bigg)^{i}
=
\partial_j(\overline{\zeta}\zeta)^{i}
+
(\overline{\zeta}\zeta)^{i}\partial_j\theta,
\end{equation} 
where the numerical factor $2$ was already absorbed into the $\theta$ function. Hence, Eq. (\ref{dpin}) amounts out to 
\begin{eqnarray}
d\pi^{i}=\sum_j\Big(\partial_j(\overline{\zeta}\zeta)^{i}+(\overline{\zeta}\zeta)^{i}\partial_j\theta\Big)dx^j=dx^{i}+x^{i}d\theta.
\end{eqnarray}

Inserting this last expression back into (\ref{dff}), we are left with 
\begin{equation}\label{df corrigido}
df(x_{0})
=
\sum\limits_{i}
\frac{\partial f}{\partial x^{i}}(x_{0})dx^{i}
+
\sum\limits_{i}\frac{\partial f}{\partial x^{i}}(x_{0})x^{i}d\theta.
\end{equation} 
The first term in (\ref{df corrigido}) is the usual one, while the second term is a direct consequence of nontrivial topology. We shall keep the standard basis for the dual $(\widetilde{\mathbb{R}^n})^*$ space and see how $\theta$ terms impact the coefficients. Here we further notice that ultimately $\partial\theta$ corrections take place in complex functions, coordinated by $\{x^\mu\}$. Therefore, when understood under the point of view of a covariant vector field, it is conceivable to write $d\theta$ in the coordinate basis $d\theta=\sum_{j=1}^n\partial_j\theta dx^j$. Thus, fairly direct exercise leads (\ref{df corrigido}) to 
\begin{equation}\label{dep}
df(x_{0})
=
\sum\limits_{i}
\bigg\{
\frac{\partial f}{\partial x^{i}}(x_{0})
+
\frac{\partial \theta}{\partial x^{i}}\big(
\vec{\nabla}f(x_{0})\cdot\vec{x}
\big)
\bigg\}dx^{i}.
\end{equation} 

Notice that, despite $df$ being calculated in $x_0\in U$, there is explicit dependence on the coordinates. We shall discuss this point, along with the variation of $\theta$, in the following. Firstly, we emphasize that if the topology is trivial, or $\partial\theta$ can be ignored, then the usual case is recovered, and no topological ``dilatation'' is found. In a generic vector, say $\phi$, this dilatation also acts 
\begin{equation}\label{aaa}
\phi
=
\sum\limits_{i}\phi_{i}(dx^{i}+x^{i}d\theta)
=
\sum\limits_{i}\Bigg(
\phi_{i}+\Big(\sum_k\phi_k x^k\Big)\partial_{i}\theta
\Bigg)dx^{i}. 
\end{equation} This last expression is sufficient to evince the peculiarity of dealing with nontrivial topology. Consider, for example, the case in which $\phi$ would have only $z$ coordinate in the Cartesian system with trivial topology. In such context, also consider that $\theta$ is a function of $x$ only. Even in this simple case, if the topological non-triviality cannot be ignored (non-negligible $\partial \theta(x)/\partial x\equiv \theta'(x)$), the vector $\phi$ acquires an extra component $\phi=\phi_zdz+\phi_zz\theta'(x)dx$. 

Before going further, it is important to revisit a well-known standard result in linear algebra, calling attention to a peculiarity in the formulation presented here.
\begin{prop}\label{prop3.1}
Let $\{e_i\}_{i=1\cdots n}$ be a base of a real vector space $\widetilde{\mathbb{R}^n}$ of dimension $n$. The linear applications $\{\varepsilon^i\}_{i=1,\cdots,n}$ such that $\varepsilon^{i}(e_{j})=\delta^{i}_{\;j}$ (as $dx^i$) are base of $(\widetilde{\mathbb{R}^n})^*$, provided the nontrivial topology is such that $\partial_{j}\theta\neq\frac{-\phi_j}{\sum_k\phi_k x^k}$ (for $\phi\neq 0$).
\end{prop}
\begin{proof}
Broadly, the proof is usually found in linear algebra textbooks: 1) if $\sum\limits_{i}\alpha_{i}\varepsilon^{i}=0$ then $\sum\limits_{i}\alpha_{i}\varepsilon^{i}(v)=0$ $\forall v\in \widetilde{\mathbb{R}^n}$; particularly when $v=e_j$ we are lead to $\alpha_j=0$ and the independence linear is assured. 2) As $\phi(v)=\sum_i\alpha_i\varepsilon^i(\sum_jv^je_j)=\sum_i\alpha_iv^i$ and $\varepsilon^j(v)=v^j$, we have $	\phi(v)=\sum\limits_{i}\alpha_{i}\varepsilon_{i}(v)$ for all $v$, that is $(\widetilde{\mathbb{R}^n})^*$ is generated by $\{\varepsilon_{i}\}$.

The novelty in this formulation is the specificity arising in part 1), since the coefficient $\alpha_j$ is given by (\ref{aaa}) $\alpha_{j}=\phi_{j}+(\sum_k\phi_k x^k)\partial_{j}\theta$. As linear independence requires $\alpha_j=0$ $\forall j$, in order not to restrict the topology, it is necessary and sufficient that  $\partial_{j}\theta\neq\frac{-\phi_j}{\sum_k\phi_k x^k}$ (for $\phi\neq 0$). Therefore $\phi_{j}+(\sum_k\phi_k x^k)\partial_{j}\theta=0$ implies $\phi_j=0$ $\forall j$, just the usual condition for linear independence. Finally, if $\sum_k\phi_k x^k=0$, then $\phi_j=0$ trivially.  
\end{proof}	   

The inner product is also peculiar, since 
\begin{equation}\label{pd}
\phi(v)=\sum_k \phi^k v_k+\bigg(\sum_i \phi^i x_i\bigg)\bigg(\sum_j \partial^j\theta v_j\bigg),
\end{equation} also reveals the influence of topological dilatation. In particular, note that when $\phi$ and $x$ are orthogonal, the usual case is reached even with nontrivial topology. This fact suggests the decomposition of $(\widetilde{\mathbb{R}^n})^*$ as 
\begin{equation}\label{dec}
(\widetilde{\mathbb{R}^{n}})^{*}
=
(\mathbb{R}^{n})^{*}_{\perp}
\oplus
(\widetilde{\mathbb{R}^{n}})^{*}\backslash(\mathbb{R}^{n})^{*}_{\perp},
\end{equation} where $(\mathbb{R}^{n})^{*}_{\perp}$ is the vector space encompassing $\{\phi\}$ such that $\sum_i\phi_i x^i=0$. Unaffected, therefore, by the nontrivial topology. 

After seeing the consequences in vector and dual spaces, it is relevant to further investigate an eventual connection between them by inspecting the metric within this context.

\subsection{The bilinear form}  

From this section on, we shall use Einstein's notation explicitly. The metric $\tilde{\eta}$ of $(\widetilde{\mathbb{R}^{1+3}})^*\otimes (\widetilde{\mathbb{R}^{1+3}})^*$ reads 
\begin{equation}\label{metr}
\tilde{\eta}=\eta_{\mu\nu}(dx^{\mu}+x^{\mu}d\theta)\otimes(dx^{\nu}+x^{\nu}d\theta),
\end{equation} where $\eta_{\mu\nu}$ is the usual Minkowski metric. It allows treating the nontrivial topology effects as corrections upon the usual case. By linearity of the tensor product, Eq. (\ref{metr}) amounts out to 
\begin{equation}\label{(2)}
\tilde{\eta}=
\eta_{\mu\nu}dx^{\mu}\otimes dx^{\nu}
+
\eta_{\mu\nu}x^{\nu}\partial_{\beta}\theta dx^{\mu}\otimes dx^{\beta}
+
\eta_{\mu\nu}x^{\mu}\partial_{\alpha}\theta dx^{\alpha}\otimes dx^{\nu}
+
\eta_{\mu\nu}x^{\mu}x^{\nu}\partial_{\alpha}\theta\partial_{\beta}\theta dx^{\alpha}\otimes dx^{\beta}.
\end{equation} Again, as the tensorial product is the standard one, $\tilde{\eta}$ is bilinear. We shall raise and lower indexes with the standard Minkowski metric. More than convenience, to serve as an isomorphism, $\tilde{\eta}$ would have to be non-degenerate, a property that is not always guaranteed, as we shall see in a moment.

\begin{prop}\label{propo 2}   
The bilinear form $\tilde{\eta}$ given by Eq. (\ref{metr}) is symmetric. 
\end{prop}
\begin{proof}
The form $\tilde{\eta}:\widetilde{\mathbb{R}^{1+3}}\times \widetilde{\mathbb{R}^{1+3}}\rightarrow \mathbb{R}$ when acting upon $v=v^\mu e_\mu$ and $\omega=\omega^\nu e_\nu$ (vectors of $\widetilde{\mathbb{R}^{1+3}}$) gives 
	\begin{multline}\label{}
	\tilde{\eta}(v,\,\omega)=
	\eta_{\mu\nu}dx^{\mu}\otimes dx^{\nu}(v^{\sigma}e_{\sigma},\,\omega^{k}e_{k})
	+
	\eta_{\mu\nu}x^{\nu}\partial_{\beta}\theta dx^{\mu}\otimes dx^{\beta}(v^{\sigma}e_{\sigma},\,\omega^{k}e_{k})
	+\\+
	\eta_{\mu\nu}x^{\mu}\partial_{\alpha}\theta dx^{\alpha}\otimes dx^{\nu}(v^{\sigma}e_{\sigma},\,\omega^{k}e_{k})
	+
	\eta_{\mu\nu}x^{\mu}x^{\nu}\partial_{\alpha}\theta\partial_{\beta}\theta dx^{\alpha}\otimes dx^{\beta}(v^{\sigma}e_{\sigma},\,\omega^{k}e_{k}),
	\end{multline} which results in 
	\begin{equation}
	\tilde{\eta}(v,\,\omega)=
	v^\mu\omega_\mu
	+
	(x^\mu v_\mu)(\partial_\beta\theta\omega^\beta)
	+
	(x^\mu\omega_\mu)(\partial_\beta\theta v^\beta)
	+
	x^{2}(\partial_\beta\theta\omega^\beta)(\partial_\kappa\theta v^\kappa)
	=
	\widetilde{\eta}(\omega,\,v).
	\end{equation}
\end{proof}	

In order to discuss non-degeneracy, let us compute $\tilde{\eta}(v^{\sigma}e_{\sigma},\,\omega)$ for an arbitrary vector $\omega\in \widetilde{\mathbb{R}^{1+3}}$:
\begin{equation}\label{o1}
\widetilde{\eta}(v^{\sigma}e_{\sigma},\,\omega)
=
\eta_{\mu\nu}
\bigg\{
(v^{\mu}+x^{\mu}\partial_{\sigma}\theta v^{\sigma})
dx^{\nu}
+
x^{\nu}\partial_{\beta}\theta
(v^{\mu}+x^{\mu}\partial_{\sigma}\theta v^{\sigma})
dx^{\beta}
\bigg\}(\omega),
\end{equation} or in a more compact form
\begin{equation}\label{o2}
\widetilde{\eta}(v^{\sigma}e_{\sigma},\,\omega)
=
(v_{\nu}+x_{\nu}\partial_{\sigma}\theta v^{\sigma})
\big(
dx^{\nu}+x^{\nu}d\theta
\big)(\omega).
\end{equation} 
The non-degeneracy condition states that if $\eta(v,\,\omega)=0, \;\forall\omega\in\widetilde{\mathbb{R}^{1+3}}$ then necessarily $v=0$. From (\ref{o2}) we see that if $\tilde{\eta}(v,\,\omega)=0,\;\forall\omega\in\widetilde{\mathbb{R}^{1+3}}$, then $v_{\nu}+x_{\nu}\partial_{\sigma}\theta v^{\sigma}=0$. Certainly, when $v=0$, the usual condition is reached (here as a sufficiency), but it is not necessary. Of course, one could add another constraint to the topological scenarios and work with $\tilde{\eta}$ non-degenerate. It may be assumed as the necessity appears (see the Appendix section). For a general discussion, we just reinforce that if $v_{\nu}+x_{\nu}\partial_{\sigma}\theta v^{\sigma}=0$, then $\tilde{\eta}$ is degenerate and some vectors in $\widetilde{\mathbb{R}^{n}}$ may not have counterpart in $(\widetilde{\mathbb{R}^{n}})^*$. Therefore we shall keep the conservative approach of using the Minkowski metric as the isomorphic bridge between these two vector spaces.  

Let us now explore the special case $\tilde{\eta}(v,v)$. As it can be readily verified from (\ref{(2)}), we have 
\begin{equation}\label{o3}
\tilde{\eta}(v,\,v)=
v^{2}+x_{\mu}v^{\mu}\partial_{k}\theta v^{k}
+
x_{\nu}v^{\nu}\partial_{\sigma}\theta v^{\sigma}+x^{2}\partial_{\sigma}\theta v^{\sigma}\partial_{k}\theta v^{k},
\end{equation} or in a more familiar form
\begin{equation}
\tilde{\eta}(v,\, v)=(v^\mu+x^\mu\partial_\alpha\theta v^\alpha)(v_\mu+x_\mu\partial_\alpha\theta v^\alpha)
\end{equation} and, in light of the above discussion, there are no light-like vectors for a degenerate $\tilde{\eta}$. For a non-degenerate $\tilde{\eta}$, however, there are vectors for which $\tilde{\eta}(v,\, v)=0$. Let us explore this case further by taking advantage of a simplified scheme where $v^\mu=\Delta x^\mu=(c\Delta t,-\Delta x)$ denotes spacetime displacements in two dimensions, and exceptionally for this example, we do not take natural units. A direct evaluation of the metric leads to 
\begin{equation}\label{ex1}
\widetilde{\Delta s}^2=\Delta x^\mu\Delta x_\mu+2x_\mu \Delta x^\mu \partial_\alpha\theta\Delta x^\alpha+ x_\mu x^\mu(\partial_\alpha\theta\Delta x^\alpha)^2.
\end{equation} So far, we have made no considerations about the nontrivial topology. This generality plays a hole here: it stands for a comprehensive analysis and highlights effects to occur independently of the class of nontrivial topology. Nevertheless, this concept could also be benefited from a more systematic approach, borrowing concepts from differential geometry as the geometrical characterization of horizons \cite{bhs}, for instance. Heretofore, we shall explore the case in which it (and its effects) are localized in a given finite region, say $\mathcal{R}$. In such context, it is generally expected that the nontrivial topology affects physical systems in a given neighborhood $\mathcal{V}$ of $\mathcal{R}$ ($\mathcal{V}\supset\mathcal{R}$) so that the net effect perturbs the usual case in $\mathcal{V}\backslash\mathcal{R}$ but may be neglected in other domains of the spacetime. While we shall deal with the region of analysis in a moment, for a typical and concrete example for this general discussion, one could understand $\mathcal{R}\simeq \mathbb{R}\times(\mathbb{R}^2\times S^1)$ and $\mathcal{V}\backslash\mathcal{R}\simeq \widetilde{\mathbb{R}^{1+3}}$. Back to our discussion, simple requirements on the function encoding nontrivial topology effects may implement this locality aspect. We consider that, in the suitable region within $\widetilde{\mathbb{R}^{1+3}}$, the first $\theta(x)$ derivative shall be small. From this perspective, we can neglect quadratic (even mixed) derivative terms in this example, and Eq. (\ref{ex1}) reads 
    
\begin{equation}\label{ex2}
\widetilde{\Delta s}^2=c^2\Delta t^2\bigg\{1+2t\dot{\theta}-\frac{u^2}{c^2}(1-2x\theta')-\frac{2u}{c}(t\theta'c^2+\dot{\theta}x)\bigg\}, 
\end{equation} where $u\equiv\Delta x/\Delta t$ is the displacement velocity, $\dot{\theta}$ denotes time derivative of $\theta$ and $\theta'$ stands for its derivative with respect to $x$. From (\ref{ex2}) we see that $\tilde{\eta}(v,v)=\tilde{ds}^2=0$ may be solved for $u$ leading to\footnote{In the following subsection, we elaborate on the spacetime region for which the approximations here performed may be safely taken. It is done to the case at hand by considering a spacetime region such that $||r||_{\max}\partial\theta \ll 1$, $r=t,x$, for constant $\partial\theta$.} 
\begin{equation}
u\approx \pm c-\dot{\theta}(x\mp ct)-c\theta'(ct\mp x).
\end{equation} This last expression, even for this simplified case, is still interesting for two related points: the ``light-like'' case here does not mean displacements at the velocity of light. The light cone is disturbed due to topological effects. However, when $\dot{\theta}\sim 0 \sim \theta'$ and the topology is (or maybe treated as) trivial, $u=\pm c$ as expected.

Returning to a less prosaic notation, from (\ref{o3}), we remark that even in the case that $v^2=0$ (a light-like vector in the standard, trivial topology, nomenclature), we are left with 
\begin{equation}\label{o4}
\tilde{\eta}(v,\,v)=2(x_\mu v^\mu)(\partial_\alpha\theta v^\alpha)+x^{2}(\partial_\alpha\theta v^\alpha)^{2}
\end{equation} and not imposing additional constraints in the topology, there are exceptional cases for which a standard light-like vector $v$ is also a light-like vector within the context of nontrivial topology, namely: i) for vectors such that $v_\alpha\partial^\alpha\theta=0$ or ii) $x_\alpha v^\alpha=0$ and $x$ is itself light-like in the standard nomenclature. Notice that in the less restrictive case i), the set of vectors belongs to $(\mathbb{R}^{n})_{\perp}$ (see discussion around (\ref{dec})), as expected. 

At this point, it is clear that usual spacetime transformations are not in general symmetries of $\widetilde{\mathbb{R}^{1+3}}$, since the presence of $\theta$ terms in the metric jeopardizes the very idea of inertial frames. However, some discussion is in order. From (\ref{(2)}), notice that, acting upon a basis, $\tilde{\eta}$ gives
\begin{equation}\label{base}
\tilde{\eta}(e_{\alpha},\,e_{\beta})=
\eta_{\alpha\beta}+x_{\alpha}\partial_{\beta}\theta+x_{\beta}\partial_{\alpha}\theta+x^{2}\partial_{\alpha}\theta\partial_{\beta}\theta\equiv\widetilde{\eta}_{\alpha\beta},
\end{equation} from which the simple form $\tilde{\eta}=\tilde{\eta}_{\alpha\beta}dx^\alpha\otimes dx^\beta$ is reached. The derivative of $\theta$ entering into $\tilde{\eta}_{\alpha\beta}$ would naturally lead to the appreciation of a curved, not flat, spacetime. While pursuing this line of research seems feasible, we shall postpone it to future work (apart from a comment in the Conclusion section). By now, to further explore flat spacetime consequences of this construction, we shall adopt the following particularization: we said that, in general, we are treating $\partial\theta$ as an allegedly small correction so that we may disregard $(\partial\theta)(\partial\theta)$ terms. Moreover, we shall also consider well-behaved corrections, that is, variations of $\theta$ which do not vary itself appreciably, a condition whose mathematical implementation reads $\partial^2\theta\rightarrow 0$ in all the relevant regions of interest. Within this context, there is no curvature associated with $\tilde{\eta}_{\mu\nu}$ in a torsionless, metric-compatible setup. We are aware that this may sound like an oversimplification. However, we shall pursue in this paper flat spacetime consequences of the construction\footnote{The metric-compatible condition is indeed trivially satisfied within this set of assumed simplifications.}. 

In light of the previous discussion, we shall deal with 
\begin{equation}
\tilde{\eta}_{\alpha\beta}=\eta_{\alpha\beta}+x_{\alpha}\partial_{\beta}\theta+x_{\beta}\partial_{\alpha}\theta, \label{sim}
\end{equation} with inverse given, then, by $\tilde{\eta}^{\alpha\beta}=\eta^{\alpha\beta}-x^{\alpha}\partial^{\beta}\theta-x^{\beta}\partial^{\alpha}\theta$. It is clear that a given theory constructed upon $\widetilde{\mathbb{R}^{1+3}}$, given the explicit $x^\mu$ terms, violate Lorentz symmetries. It makes the construction here developed a natural scenario for Lorentz violating models \cite{VL}, and the vast literature pointing to upper limits to the Lorentz violation terms could be used here to constrain $\partial\theta$ terms in principle. Besides that, it is quite clear, from (\ref{sim}), that the invariance under spacetime translations is, at most, only approximate here. Actually, $\widetilde{\mathbb{R}^{1+3}}$ is not an affine space due to the localized nontrivial topology presence.    

\subsection{Dynamics in $\widetilde{\mathbb{R}^{1+3}}$: quasinormal-like modes}

The nontrivial spacetime topology has induced corrections on the metric; hence, the Clifford mapping shall also reflect this feature. To appreciate that, we may define tetrad-like objects $e^\mu_{\;\,\alpha}=\delta^\mu_\alpha-x^\mu\partial_\alpha\theta$ and $e_\mu^{\;\,\alpha}=\delta_\mu^\alpha+x^\alpha\partial_\mu\theta$ such that $e^\mu_{\;\,\alpha}e_\mu^{\;\,\beta}=\delta_\alpha^\beta$ and $e^\mu_{\;\,\alpha}e_\nu^{\;\,\alpha}=\delta^\mu_\nu$ up to first order derivatives in $\theta$. It can be readily verified that 
\begin{equation}
\tilde{\eta}^{\mu\nu}=e^\mu_{\;\,\alpha}e^\nu_{\;\,\beta}\;\eta^{\alpha\beta} \label{t3} 
\end{equation} and, of course, $\tilde{\eta}_{\mu\nu}=e_\mu^{\;\,\alpha}e_\nu^{\;\,\beta}\eta_{\alpha\beta}$. Hence, by defining $\tilde{\gamma}^\mu=e^\mu_{\;\,\alpha}\gamma^\alpha$, we have 
\begin{equation} 
\{\tilde{\gamma}^\mu,\tilde{\gamma}^\nu\}=\{\gamma^\mu,\gamma^\nu\}-x^\nu\partial_\beta\theta\{\gamma^\mu,\gamma^\beta\}-x^\mu\partial_\alpha\theta\{\gamma^\alpha,\gamma^\nu\}, \label{t4}
\end{equation} which, by using the standard Clifford algebra relation for $\gamma^\mu$, leads to its familiar form for $\tilde{\gamma}^\mu$, that is $\{\tilde{\gamma}^\mu,\tilde{\gamma}^\nu\}=2\tilde{\eta}^{\mu\nu}\mathbb{1}$, where $\mathbb{1}$ stands for the identity. Similarly, for $\tilde{\gamma}_\kappa=e_\kappa^{\;\,\alpha}\gamma_\alpha$, we have $\{\tilde{\gamma}_\mu,\tilde{\gamma}_\nu\}=2\tilde{\eta}_{\mu\nu}\mathbb{1}$. The recovery of the Clifford relation in terms of $\tilde{\gamma}^\mu$ gives a clue for the corrected Dirac operator in $\widetilde{\mathbb{R}^{1+3}}$. Thus, fermions in this spacetime shall be subject to dynamics given by $(i\tilde{\gamma}^\mu\partial_\mu-m)\Psi=0$. A remark must be made after such a dynamic equation has been written. As a result of the spacetime deformation to $\widetilde{\mathbb{R}^{1+3}}$, one should be careful about the type (so to speak) of the resulting topology. There are nontrivial examples of constructed base manifolds whose net result is the impossibility of the existence of a spinor structure. Here we shall assume a suitable nontrivial topological deformation so that non-unique spinor structures occur\footnote{In this regard, a classic theorem by Geroch (first paper of \cite{exot}) stays, roughly speaking, that a spinor structure is guaranteed in a non-compact manifold $\mathcal{M}$ as far as an Eq. like (\ref{t3}) is assured at every point of $\mathcal{M}$.}, but we call attention to the fact that, as studied in the Appendix, some standard concepts regarding exterior and differential forms are quite the same over $\widetilde{\mathbb{R}^{1+3}}$. In contrast, other concepts need additional care and further investigation. Even in the case in which standard theorems apply, however, as the symmetry group of $\widetilde{\mathbb{R}^{1+3}}$ is not the Lorentz group in general, a spinor here does not mean a mathematical object carrying a spin $1/2$ representation of the Lorentz group but a deformed object which recovers its usual concept of spinor as $\theta\rightarrow cte$, i. e. when the topology is trivial.

Analyzing the Klein-Gordon-like equation emerging from squaring the Dirac-like dynamical equations is informative. Squaring the exotic spinor equation in $\widetilde{\mathbb{R}^{1+3}}$ is work out the expression
\begin{equation}
(i\tilde{\gamma}^\nu\partial_\nu+m)(i\tilde{\gamma}^\mu\partial_\mu-m)\Psi=0. \label{q1}
\end{equation} Bearing in mind the discussion around (\ref{base}) we have $\tilde{\gamma}^\mu\partial_\mu\theta\simeq \gamma^\mu\partial_\mu\theta$. Moreover, it is convenient to make all the gamma tilde terms explicit since it has non-vanishing derivatives due to the presence of $x^\mu$. In the process of squaring, even within our approximations, terms as $\gamma^\nu\partial_\nu(\gamma^\alpha\partial_\alpha\theta\Psi)\simeq \gamma^\nu\gamma^\alpha\partial_\alpha\theta\partial_\nu\Psi$ shall appear. The most tricky term is given by $\gamma^\nu\partial_\nu(\gamma^\alpha\partial_\alpha\theta x^\mu\partial_\mu\Psi)$, which amounts to $\gamma^\nu\gamma^\alpha\partial_\alpha\theta(\partial_\nu\Psi+x^\mu\partial_\nu\partial_\mu\Psi)$. After these considerations, we have, as a partial result,
\begin{eqnarray}
-\gamma^\nu\gamma^\mu\partial_\nu\partial_\mu\Psi\!\!\!&+&\left.\!\!\!\gamma^\nu\gamma^\alpha\partial_\alpha\theta\partial_\nu\Psi+\gamma^\nu\gamma^\alpha\partial_\alpha\theta x^\mu\partial_\nu\partial_\mu\Psi\right.\nonumber\\&+&\left.\!\!\! \gamma^\alpha\gamma^\mu\partial_\alpha\theta x^\nu\partial_\nu\partial_\mu\Psi-m^2\Psi=0. \right.\label{f1}
\end{eqnarray} 

Now, by means of the standard relation $\gamma^\mu\gamma^\nu=\frac{1}{2}[\gamma^\mu,\gamma^\nu]+\eta^{\mu\nu}$ all the remain steps follow straightforwardly to 
\begin{equation}
(\Box+m^2)\Psi-\partial^\alpha\theta(1+2x^\mu\partial_\mu)\partial_\alpha\Psi+\frac{1}{2}
[\gamma^\alpha,\gamma^\beta]
\color{black}
\partial_\alpha\theta\partial_\beta\Psi=0\label{f2}
\end{equation} and in the case of trivial topology (or topological negligible net effect), the usual Klein-Gordon is recovered, as expected. Eq. (\ref{f2}) has significant consequences to be highlighted. Notice the existence of a `dissipative' first derivative term as a consequence of a nontrivial topology. This factor shall impact the dispersion relation resulting in a term leading to the quasinormal-like modes. Quasinormal modes naturally arise in linear black hole perturbation theory \cite{quasi, Berti:2022hwx}. It describes the so-called ``ringdown'' phase of the black hole coalescence, corresponding to the final part of the gravitational waves signal, as in recent LIGO/VIRGO detections \cite{LIGOScientific:2016aoc, LIGOScientific:2017ycc, LIGOScientific:2019fpa, LIGOScientific:2018dkp}. When perturbing a vibrating string, the boundary conditions select discrete values of the frequency $\omega$, as well as the normal modes of the system. Although their existence also depends on the boundary conditions, quasinormal frequencies are given by complex numbers. The real part is the proper frequency, and the imaginary part is associated with the decay timescale of the dumped oscillation. Quasinormal modes are indeed valuable for treating the dissipative feature of the system. In fact, beyond the usual simplifications, realistic physical systems are dissipative. Thus it would be reasonably expected that quasinormal modes appear in a broad class of problems in physics.

It is worth noticing that bosonic fields shall also respect Eq. (\ref{f2}). It is a strength of the formalism we investigate here: nontrivial topology impacts all physical fields. Therefore, as a program to extract physical information about Eq. (\ref{f2}), let $\phi(x)$ be a scalar field and denote by $v^\alpha\equiv\partial^\alpha\theta$ the constant (in this approximation) vector encompassing topological effects. We shall study the dispersion relation coming from the matrix equation
\begin{equation}
\big\{(\Box+m^2)\phi(x)-v^\alpha\partial_\alpha\phi(x)-2v^\alpha x^\mu\partial_\mu\partial_\alpha\phi(x)\big\}\mathbb{1}+\frac{1}{2}
[\gamma^\alpha,\gamma^\beta]
\color{black}
\partial_\alpha\theta\partial_\beta\phi=0.\label{p1}
\end{equation} It is necessary to take some caution in handling the Fourier transform of (\ref{p1}) due to the presence of $x^\mu$. In general, a Fourier transform is well defined if the function to be transformed is integrable in any finite region, which is the case for $x^\mu$. However, the function should also vanish at infinity, which is problematic. We note once again that for Eq. (\ref{p1}) being valid, we are necessarily in a region of spacetime where the influence of the nontrivial topology is small enough that $\partial\theta$ is considered a constant.
%
%
\color{black}
Far away from this region, even $\partial\theta$ shall vanish. Therefore, we shall argue about integrability conditions within the suitable region for which Eq. (\ref{p1}) makes sense. The following construction may systematize these considerations: let $\widetilde{\mathcal{U}}\subset \widetilde{\mathbb{R}^{1+3}}$ to contain a nontrivial topology delimited in a finite sub-region of $\widetilde{\mathcal{U}}$, so that there exists $\mathcal{U}\subset\widetilde{\mathcal{U}}$ locally isomorphic to $\mathbb{R}^{1+3}$; that is, $\mathcal{U}$ is endowed with trivial topology. We shall consider $\mathcal{U}$ distancing what is necessary from the localized nontrivial topology so that the applied approximations are still valid. Consider also a finite time interval, say $\mathcal{I}$, and denote by $\mathcal{E}\subset \mathbb{R}^{3}$ a finite space region. We shall perform the Fourier transform in $\Sigma\equiv (\mathcal{I}\times\mathcal{E})\subset\mathcal{U}$ assumed then compact and orientable. Besides, we use $\langle B \rangle|_{\partial\Sigma}$ as a notation for boundary terms. Hence, if  
\begin{equation}
\mathcal{F}[\phi(x)](p)=\int_{\Sigma}d^4x \,\phi(x)e^{ipx} \label{p2}
\end{equation} is the Fourier transform of $\phi(x)$ in $\Sigma$, then it readily follows  
\begin{equation}
\mathcal{F}[\partial_\alpha\phi(x)](p)=-i\int_{\Sigma}d^4x \; p_\alpha \phi(x)e^{ipx}+\langle \phi(x)e^{ipx}\rangle|_{\partial\Sigma_\alpha},\label{p3} 
\end{equation}
\begin{equation}
\mathcal{F}[\partial_\mu\partial_\alpha\phi(x)](p)=-\int_{\Sigma}d^4x \; p_\mu p_\alpha \phi(x)e^{ipx}+\langle \partial_\alpha\phi(x)e^{ipx}\rangle|_{\partial\Sigma_\mu}-i\langle p_\mu \phi(x)e^{ipx}\rangle|_{\partial\Sigma_\alpha},\label{p4} 
\end{equation} where $\partial\Sigma_\mu$ denotes the boundary of $\Sigma$ taken from integration over $\partial_\mu$. From (\ref{p4}) we get 
\begin{equation}
\mathcal{F}[\Box\phi(x)](p)=-\int_\Sigma d^4x \;p^2\phi(x)e^{ipx}+\langle \partial_\alpha\phi(x) e^{ipx} \rangle|_{\partial\Sigma^\alpha}-i \langle p^\alpha \phi(x) e^{ipx} \rangle|_{\partial\Sigma_\alpha}. \label{p5}
\end{equation} The Fourier transform of $x^\mu$ in $\Sigma$ is 
$\mathcal{F}[x^\mu](p)=\int_\Sigma d^4x \, x^\mu e^{ipx}=-i\partial_p^\mu \int_\Sigma d^4x \,e^{ipx}$, where $\partial_p^\mu$ stands for derivative with respect to the momentum. As $\Sigma$ is a finite subregion and the integral does not take in the entire space domain, we cannot associate an exact Dirac delta representation to the last integral. We shall write $\Delta(p)=\int_\Sigma d^4x \,e^{ipx}$, so that $\mathcal{F}[x^\mu](p)=-i\partial_p^\mu\Delta(p)$. The functional form of $\Delta(p)$ is relevant since it is related to the setting of $\Sigma$. 

The last point to be stressed is that the last term in Eq. (\ref{p1}) makes necessary the use of the (converse of) convolution theorem. For ready reference, the convolution of two functions is given as usual by $f*g=\int d^4k' f(k')g(k-k')$, so that the Fourier transform of $x^\mu\partial_\mu\partial_\alpha\phi(x)$, given by the convolution of $\mathcal{F}[x^\mu]$ and (\ref{p4}), reads\footnote{Assuming we can interchange integrals accordingly.}  
\begin{eqnarray}
\mathcal{F}[x^\mu\partial_\mu\partial_\alpha\phi(x)]&=&\left. i\int_\Sigma d^4x \int d^4p' \;\partial_{p'}^\mu\Delta(p') (p-p')_\mu (p-p')_\alpha \phi(x) e^{i(p-p')x}\right.\nonumber\\&-&\left. i\langle \partial_\alpha\phi(x)\partial_p^\mu\Delta*e^{ipx} \rangle|_{\partial\Sigma_\mu}-\langle\phi(x)\partial_p^\mu\Delta*p_\mu e^{ipx} \rangle|_{\partial\Sigma_\alpha}. \right. \label{p6}
\end{eqnarray} Taking all these considerations into account, the Fourier transform of Eq. (\ref{p1}) may be written as\footnote{Multiplication by $\mathbb{1}$ in the boundary terms, without any other matrix specification, is tacitly implied.} 
\begin{eqnarray}&&
\left. \int_\Sigma d^4x \,\phi(x) \bigg[(-p^2+m^2+iv^\alpha p_\alpha)\mathbb{1}+\frac{i}{2}[\gamma^\beta,\gamma^\alpha]v_\alpha p_\beta\bigg]e^{ipx}-i\int_\Sigma d^4x \,2v^\alpha \phi(x) \partial_p^\mu\Delta*p_\mu p_\alpha e^{ipx}\right.\nonumber\\&+&\left.\big\langle v^\alpha\phi(x)\partial_p^\mu\Delta*p_\mu e^{ipx} \big\rangle\big|_{\partial\Sigma_\alpha}+\Big\langle [\partial^\alpha\phi(x)-(v^\alpha+ip^\alpha)\phi(x)]e^{ipx}\Big\rangle\Big|_{\partial\Sigma_\alpha}+2i\big\langle v^\alpha\partial_\alpha\phi(x)\partial_p^\mu\Delta*e^{ipx}\big\rangle\big|_{\partial\Sigma_\mu}\right.\nonumber\\&+&\left.
\bigg\langle \frac{1}{2}
[\gamma^\alpha,\gamma^\beta]
\color{black}
v_\alpha \phi e^{ipx}\bigg\rangle\bigg|_{\partial\Sigma_\beta}=0.\right.\label{p7}
\end{eqnarray} The first point to emphasize about (\ref{p7}) is that for a trivial topology, i. e. $v^\alpha=0$, there is no restriction to the $\Sigma$ region, and we are left with the usual relativistic dispersion relation. In the case of nontrivial topology, in the approximation context already discussed, Eq. (\ref{p7}) may be useful in practical applications for which the $\Sigma$ region may be delimited, provided appropriate boundary conditions, and the $\Delta(p)$ function properly evaluated. Incidentally, depending on the case to be investigated, the Fourier transform might even be discrete. However, for a general analysis not specifying the $\Sigma$ region, one cannot go further without some assumption about $\Delta(p)$. It is a good point to emphasize that an inspection on Eq. (\ref{p7}) evinces the existence of quasinormal energy dissipative terms due to the $iv^\alpha p_\alpha$ term. That is to say, the nontrivial topology indeed perturbs the field modes, and damping (or even amplifying) energy terms are in order. This aspect shows the necessity of some care in handling $p_0$ variables, opening the possibility for integration in the complex plane when computing the convolutions. Here we shall adopt a simpler setup for $\Sigma$, eliminating most complications due to quasinormal modes and making its consequences explicit. In fact, $\Delta=\int_\Sigma d^4x e^{ipx}=\int_\mathcal{I} dt\, e^{ip_0t}\int_\mathcal{E}d^3x\, e^{-i{\bf p}\cdot{\bf x}}$. Opening the possibility for a complex $p_0$ integration in time leads to 
\begin{equation}
\Delta=-\frac{i}{p_0}\Big[e^{i\mathrm{Re}(p_0)t}e^{-\mathrm{Im}(p_0)t}\Big]\bigg|_{\inf{(\mathcal{I})}}^{\sup{(\mathcal{I})}}\int_\mathcal{E}d^3x\, e^{-i{\bf p}\cdot{\bf x}}.\label{p8}
\end{equation} For a suitable $\Sigma$ region, the above integration may be done very small. If it is regarded as of first $\theta$ derivatives order, then, since $\Delta$ terms always appear multiplied by $v^\alpha$ in Eq. (\ref{p7}), they may quite well be disregarded from the analysis. While this is not the most general case, there is always a $\Sigma$ region for which this is indeed the case\footnote{As we are going to see in a moment, due to an additional constraint, the condition to disregard $v^\alpha\Delta$ is not entirely input to the time integration. Moreover, we stress that if we do not consider the interval of integration (times $v^\alpha$) small, but the exponential itself when multiplied by $v^\alpha$, then, after imposing suitable boundary conditions in Eq. (\ref{p7}), we are left only with the usual dispersion relation.}, provided $\mathrm{Im}(p_0)>0$. Besides, it has the bonus of being the simplest case and suffices to illustrate the physical effects. 

We are now able to rewrite Eq. (\ref{p7}) in light of the last approximations. As the last step, the adoption of joint Neumann-Dirichlet boundary conditions eliminates all the boundary terms\footnote{This procedure is the simplest one, but it would be important to investigate possible boundary dynamics issues \cite{PP} in the formulation.}, so that we are left with a complex matrix dispersion relation given by
\begin{eqnarray}\label{Rel_Dispr}
(-p^2+m^2+iv^\alpha p_\alpha)\mathbb{1}+\frac{i}{2}
[\gamma^\alpha,\gamma^\beta]
\color{black}
v_\alpha p_\beta=0. \label{ccr}
\end{eqnarray} Certainly, these boundary conditions are not consistent with every $\Sigma$ region. Once again, we are taking for granted, perhaps too naively, its validity for some $\Sigma$. As far as we can see, all the requirements imposed over $\Sigma$ are consistent, but of course, the generality of the formulation cannot be claimed anymore. The influence of topology can be made explicit by the following reasoning: all the off-diagonal elements of (\ref{Rel_Dispr}) vanish under the additional constraints  $\frac{p_0}{v_0}=\frac{p_1}{v_1}=\frac{p_2}{v_2}=\frac{p_3}{v_3}$. These very same conditions eliminate eventual diagonal terms other than $(-p^2+m^2+iv^\alpha p_\alpha)$. Taking into account the constraints, $p_j=\frac{v_j}{v_0}E$, it is straightforward to find the complex spectrum given by 
\begin{eqnarray}
 E\approx& i\frac{\dot{\theta}}{2}\pm m\big(1-\frac{1}{2}(\nabla\theta/\dot{\theta})^2\big).\label{sp1}
\end{eqnarray} Of course, the limit for trivial topology shall be taken already in (\ref{Rel_Dispr}), leading to the standard case. Besides, the spacial momentum contributes to the quasinormal-like behavior through the additional constraints. Let us resort once again to the two-dimensional case to gain insight into the physical results. The field modes are given by $e^{ipx}=\exp{[i(\mathrm{Re}(p_0)t-\mathrm{Re}(p_x)x)]}\exp{[-\mathrm{Im}(p_0)t+\mathrm{Im}(p_x)x]}$. By inspecting once again the $\Delta$ integral, also taking into account the possibility for $p_x$ be a complex number as the additional constraints suggest, we can see what kind of quasinormal-like behavior is necessary to disregard $\Delta v^\alpha$ terms. In the extreme and simplest case $\mathrm{Im}(p_0)=\dot{\theta}/2>0$ and $\mathrm{Im}(p_x)=\theta'/2<0$ and therefore $e^{ipx}$ is given by an oscillating term multiplied by $\exp{[-(\dot{\theta}t+|\theta'|x)/2]}$. The discussion around (\ref{p8}) says that this is a minimal damping effect, and even in the case $\dot{\theta}$ or $\theta'$ have a different sign, the net effect is always of damping. The nontrivial topology affects the $e^{ipx}$ outgoing (with respect to $\Sigma$) field modes as a friction term, and the opposite interpretation applies to ongoing modes. The whole situation may be summarized as follows: when reaching a spacetime region in which nontrivial topology effects are to be felt, field modes are split into ongoing and outgoing modes. The outgoing modes suffer a friction-like behavior, while the former has a reinforcement. In this last case, the modes shall rapidly lose the approximation conditions necessary for our approach.

Bearing in mind black hole quasinormal typical analysis, one could start with $\phi(t,x) = e^{i\omega t}\varphi(x)$ and study the behavior of the field modes from the equation of motion. However, in practice, the constraint eliminating off-diagonal elements in Eq. (\ref{p1}) trivializes the problem by allowing a first derivative equation for $\varphi(x)$. In this context, Eq. (\ref{p1}) would serve as the one to be solved for $\omega$. To make contact with the usual approach, let us work out Eq. (\ref{p1}) without implementing the topological constraint for a while. Remember that $\theta$ is linear in both variables, so that $\theta(t,x)\approx \alpha t+\beta x$, where $\alpha$ and $\beta$ are real constants. Taking all that into account, Eq. (\ref{p1}) reads  
\begin{align}
-\dfrac{\partial^2\varphi}{\partial x^2} + \dfrac{\partial \varphi}{\partial x} \beta + (m^2-\omega^2+\omega\alpha)\varphi(x)=0.
\end{align}
The first derivative term may be eliminated by the change variable $\dfrac{dy}{dx} = e^{\beta x}$, from which we have
\begin{align}
\dfrac{\partial^2\varphi}{\partial y^2} + (\omega^2-\omega\alpha - m^2)e^{-2\beta x}\varphi(y)=0.
\end{align}
Expanding up the exponential, we finally get
\begin{align}
\dfrac{\partial^2\varphi}{\partial y^2} + (2\beta x(y) - 1)(\omega^2 -\omega\alpha - m^2)\varphi(y)=0.
\end{align} Similar to what happens in black hole perturbation theory, we notice that it results in a more straightforward time-independent Schr\"odinger-like equation.

The friction-like interpretation of the scalar field dynamics in a scenario of nontrivial topology is endorsed by bringing back to the configuration space the simplifications derived in the momentum space. Naturally, further analysis should be applied to have a complete picture of the fundamental spacetime effects of the nontrivial topology. 

\section{Final Remarks}

In the spirit of Cartan's spinor framework, we have performed a geometrization of topology by investigating which effects would result from relating spacetime points to exotic spinor entries. We found a perturbed bilinear form encompassing nontrivial topology effects and studied the underlying mathematics it leads to. One of the strengths of this formalism is that nontrivial topology effects are likely to be felt by any field. Moreover, as the correction also involves the appearance of explicit spacetime coordinates, Lorentz violation models could benefit from this formulation because nontrivial topology geometrization could enlighten some approaches. 

After computing the Klein-Gordon-like equation in such a spacetime, we could investigate the influence of topology in the (quasinormal) modes of a scalar field. As a result, partially evanescent (or amplified) modes are expected, depending on the type of interaction with the $\theta(x)$ function reflecting the nontrivial topology. 

As remarked in the main text, relaxing a few approximations would naturally lead to a curved spacetime where nontrivial topology would serve as a source even in the absence of matter fields. It is worthwhile to pursue this branch of research. Bearing in mind the discussion around Eq. (\ref{aaa}), it would not be a surprise if some additional impositions on $\theta$ appear to avoid some variation of closed time-like curves \cite{goed}.

We conclude these remarks by pointing out an attempt to interpret our formulation in a different, less literal, form. Instead of thinking about a localized nontrivial topology, notice that the $\tilde{\eta}_{\mu\nu}$ corrections are encoded in $\partial\theta$ terms and, as $\theta$ is an adimensional function, the corrections scales with energy in natural units. The greater the energy, the greater the effect of nontrivial topology in the physical system. On the other hand, the higher the energy, the more short distances are scrutinized. Hence, it is possible to reinterpret our approach as a first attempt to investigate the physical consequences of a spacetime whose topology, in down deep scales, is nontrivial. In this context, the performed approximations find support in the energy level probed by the physical system. Besides, the finite $\Sigma$ region in which the Fourier transform takes place is still necessary since it is related to the integration in the momenta by the convolution presented in Eq. (\ref{p7}) and it cannot run over all momenta keeping valid the approximations used.

\section*{Acknowledgments}

JMHS thanks to CNPq (grant No. 303561/2018-1) for financial support. GMCR thanks to CAPES for financial support. The authors would like to express they gratitude to Prof. Elias L. Mendon\c ca for useful conversation.  

\section*{Appendix: Differential forms upon $\widetilde{\mathbb{R}^n}$}
%
%
This appendix is devoted to the study of differential forms in $\widetilde{\mathbb{R}^n}$, along with the investigation of exterior derivative within this space, for differential operators changed due to the nontrivial topology, in accordance to what was shown in the main text.

\begin{defi}
    Let $\Lambda_{1}(\widetilde{\mathbb{R}^n})$ be the space $\widetilde{\mathbb{R}^n}$ and $\Lambda^{1}(\widetilde{\mathbb{R}^n})$ its dual, $(\widetilde{\mathbb{R}^n})^*$. Similarly, the alternating $k$-vectors, and the differential $k$-forms will be denoted, respectively, by $\Lambda_{k}(\widetilde{\mathbb{R}^n})$ and $\Lambda ^{k}(\widetilde{\mathbb{R}^n})$.
\end{defi}
The $1-$form is just given by (\ref{dep}), and to find out a general expression of a $k-$form, we start investigation in some detail $2-$forms belonging to $\Lambda^{2}(\widetilde{\mathbb{R}^n})$, after what a generalization comes straightforwardly.  

A $2-$form $\omega$ $\in \Lambda^2(\widetilde{\mathbb{R}^n})$ is given by 
\begin{equation}
    \omega=\omega_{ij}\dd x^i\wedge\dd x^j,
\end{equation}
which amount to be 
\begin{equation}
    \omega=\omega_{ij}dx^i\wedge dx^j + \omega_{ij}x^j \partial_c \theta  dx^i\wedge dx^c + \omega_{ij}x^i\partial_k\theta dx^k\wedge dx^j,
\end{equation} or simply 
\begin{equation}
    \omega=(\omega_{ij} + 2\omega_{ik}x^k \partial_j \theta)dx^i\wedge dx^j.
\end{equation}
Similarly, a given $k-$form is given by 
\begin{equation}
    \omega=(\omega_{i_1 ... i_k} + k\omega_{a i_2 ... i_k}x^a \partial_{i_1} \theta)dx^{i_1}\wedge...\wedge dx^{i_k},\label{dep2}
\end{equation}
as can be seen by induction. In fact, assuming the validity of (\ref{dep2}), a $(k+1)-$form $\omega'\in\Lambda^{k+1}(\widetilde{\mathbb{R}^n})$ will be given by  
\begin{equation}
    \omega'
    =
    (
    \omega'_{i_1...i_{k+1}}
    +
    k\omega'_{a i_2...i_{k+1}} x^a \partial_{i_1}\theta
    )
    dx^{i_1}
    \wedge...\wedge
    dx^{i_{k}}
    \wedge
    \dd x^{i_{k+1}}
\end{equation} which can be recast in 
\begin{multline}
    \omega'
    =
    \omega'_{i_1...i_{k+1}}dx^{i_1}\wedge...\wedge dx^{i_{k+1}}
    +
    kx^a\partial_{i_1}\theta\omega'_{a i_2...i_{k+1}}dx^{i_1}\wedge...\wedge dx^{i_{k+1}}
    +\\+
    \omega'_{b i_2...i_{k}a}x^{a}\partial_{i_1}\theta dx^{b}\wedge dx^{i_2}\wedge...\wedge dx^{i_{k}} \wedge dx^{i_1}.    
\end{multline}
Now, permuting the order of the indices $b$ and $a$ in $\omega'_{b i_2...i_{k}a}$, we get $(-1)^{k(k-1)}\omega'_{a i_2...i_{k}b}$ and permuting the indices $i_1$ and $b$ in $dx^{b}\wedge dx^{i_2}\wedge...\wedge dx^{i_{k}} \wedge dx^{i_1}$ we get $ (-1)^{k(k-1)}dx^{i_1}\wedge dx^{i_2}\wedge...\wedge dx^{i_{k}} \wedge dx^{b}$. Thus, when making these modifications, along with the addition of taking $b\to i_{k+1}$, we arrive at 
\begin{equation}
    \omega'
    =
    \big(
    \omega'_{i_1...i_{k+1}}
    +
    (k+1)x^a\partial_{i_1}\theta\omega'_{a i_2...i_{k+1}}
    \big)
    dx^{i_1}\!\wedge...\wedge dx^{i_{k+1}}.
\end{equation}
Taking the definition 
    $\bar{\omega}_{i_1...i_k} = \big(
    \omega_{i_1...i_{k}}
    +
    kx^a\partial_{i_1}\theta\omega_{a i_2...i_{k}}
    \big)$, the above equation for the $k-$form can be rewritten in the compact form 
$\omega
    =
    \bar{\omega}_{i_1...i_k}
    dx^{i_1}\!\wedge...\wedge dx^{i_{k}}
$.

The exterior derivative in $\widetilde{\mathbb{R}^n}$ of one $k-$form, $\omega$, is given by
\begin{equation}
    \dd\omega
    =
    \partial_n
    \bar{\omega}_{i_1...i_k}
    \dd x^{n}\wedge dx^{i_1}\wedge...\wedge dx^{i_k}.
\end{equation}
which reads 
\begin{equation}
    \dd\omega
    =
    (
    \partial_n\bar{\omega}_{i_1...i_k}
    +
    x^{m}\partial_{n}\theta\partial_m\bar{\omega}_{i_1...i_k}
    )
    dx^{n}\wedge dx^{i_1}\wedge...\wedge dx^{i_k}.
\end{equation}
Since 
$x^{m}\partial_{n}\theta\partial_m\bar{\omega}_{i_1...i_k}
    =
    x^{m}\partial_{n}\theta\partial_m\omega_{i_1...i_k}
$ and 
$    \partial_{n}\bar{\omega}_{i_1...i_k}
    =
    \partial_{n}\omega_{i_1...i_k}
    +
    k\partial_{i_1}\theta\omega_{n i_2...i_k}
    +
    kx^{a}\partial_{i_1}\theta\partial_{n}\omega_{a i_2...i_k}
$ we are left with 
\begin{equation}
    \dd\omega
    =
    (
    \partial_{n}\omega_{i_1...i_k}
    +
    k\partial_{i_1}\theta\omega_{n i_2...i_k}
    +
    kx^{a}\partial_{i_1}\theta\partial_{n}\omega_{a i_2...i_k}
    +
    x^{a}\partial_{n}\theta\partial_a\omega_{i_1...i_k}
    )
    dx^{n}\wedge dx^{i_1}\wedge...\wedge dx^{i_k}.
\end{equation}

The reader certainly sees that all the usual expressions for differential forms and exterior derivatives are trivially recovered in the trivial topology limit. The differences evinced for nontrivial topology, however, bring a significant consequence. Let us investigate the second exterior derivative, which is always null in the usual case. To do so, from simplicity, take $(
    \partial_{n}\omega_{i_1...i_k}
    +
    k\partial_{i_1}\theta\omega_{n i_2...i_k}
    +
    kx^{a}\partial_{i_1}\theta\partial_{n}\omega_{a i_2...i_k}
    +
    x^{a}\partial_{n}\theta\partial_a\omega_{i_1...i_k}
    )=\{...\}_{ni_1...i_k}$. With the aid of this short notation, we have 
\begin{equation}
    \dd^{2}\omega
    =
    \partial_{m}
    \{...\}_{ni_1...i_k}
    \dd x^{m}\wedge dx^{n}\wedge dx^{i_1}\wedge...\wedge dx^{i_k},
\end{equation} resulting in 
\begin{equation}
    \dd^{2}\omega
    =
    (
    \partial_{m}
    \{...\}_{ni_1...i_k}
    +
    x^{b}\partial_{m}\theta\partial_{b}\{...\}_{ni_1...i_k}
    )
    dx^{m}\wedge dx^{n}\wedge dx^{i_1}\wedge...\wedge dx^{i_k}.
\end{equation}
Opening the term $\partial_{m}\{...\}_{ni_1...i_k}$  we are still left with the following non-vanishing result
\begin{equation}
    \dd^{2}\omega
    =
    \partial_{n}\theta\partial_{m}\omega_{i_1...i_k}
    dx^{m}\wedge dx^{n}\wedge dx^{i_1}\wedge...\wedge dx^{i_k}.
\end{equation} This expression suggests that the very definition of exotic and closed forms shall be reformulated, with possible interesting consequences in algebraic topology results. We notice, by passing, that the third exterior derivative reads
$\dd^{3}\omega
    =
    x^{a}\partial_{b}\theta\partial_{a}\partial_{n}\theta\partial_{m}\omega_{i_1...i_k}
    dx^{b}\wedge dx^{m}\wedge dx^{n}\wedge dx^{i_1}\wedge...\wedge dx^{i_k},
$
which vanishes if one is willing to accept, in this more formal context, the physical approximation we have used along the main text. 

Apart from the aforementioned sharp difference between the usual case and the case at hand, there are also similarities worth analyzing for bookkeeping purposes. We shall investigate two relevant results which are the same in both cases. Firstly, let $\omega$ and $\eta$ be a $k$ and $l-$form, respectively. Hence 
\begin{align}
    \omega&
    =
    (
    \omega_{i_1...i_k}
    +
    kx^{a}\partial_{i_1}\theta\omega_{ai_2...i_k}
    )dx^{i_1}\wedge...\wedge dx^{i_k},
    \\
    \eta&
    =
    (
    \eta_{j_1...j_l}
    +
    lx^{b}\partial_{j_1}\theta\omega_{bj_2...j_l}
    )dx^{j_1}\wedge...\wedge dx^{j_l}.
\end{align}
The exterior derivative of the product $\omega\wedge\eta$ is given by 
\begin{multline}
    \dd(\omega\wedge\eta)
    =
    \partial_{m}(
    \omega_{i_1...i_k}
    +
    kx^{a}\partial_{i_1}\theta\omega_{ai_2...i_k}
    )
    (
    \eta_{j_1...j_l}
    +
    lx^{b}\partial_{j_1}\theta\omega_{bj_2...j_l}
    )
    \dd x^{m}\wedge dx^{i_1}\wedge...\wedge dx^{i_k}\wedge dx^{j_1}\wedge...\wedge dx^{j_l}
    +\\+(-1)^{k}
    (
    \omega_{i_1...i_k}
    +
    kx^{a}\partial_{i_1}\theta\omega_{ai_2...i_k}
    )
    \partial_{m}
    (
    \eta_{j_1...j_l}
    +
    lx^{b}\partial_{j_1}\theta\omega_{bj_2...j_l}
    )
    dx^{i_1}\wedge...\wedge dx^{i_k}\wedge \dd x^{m}\wedge dx^{j_1}\wedge...\wedge dx^{j_l}
\end{multline}
and, thus, one sees that the standard relation for Leibniz's exterior derivative \cite{che} rule is recovered, that is
\begin{equation}
    \dd(\omega\wedge\eta)
    =
    \dd\omega\wedge\eta
    +
    (-1)^{\deg(\omega)}
    \omega\wedge\dd\eta.
\end{equation}

To address the second relevant similar result, let us recall some definitions designed, so to speak, for our case.  
\begin{defi}
    Let 
    $i_\lambda:\widetilde{\mathbb{R}^n}\hookrightarrow\mathbb{R}\times \widetilde{\mathbb{R}^n}$ such that $i_\lambda(x^{i_1},...,x^{i_k})\mapsto(\lambda,x^{i_1},...,x^{i_k})$ be the \textit{natural injection}.
\end{defi}
The natural injection may be, of course, defined for an open set $\widetilde{A}\subset\widetilde{\mathbb{R}^{n}}$. Besides, although $i_\lambda$ is defined over $\widetilde{\mathbb{R}^n}$, the same idea of inclusion can be extended for $\Lambda_k(\widetilde{\mathbb{R}^{n}})$ and $\Lambda^k(\widetilde{\mathbb{R}^{n}})$, so that it is also possible to engender the pullback of $\omega$ by $i_\lambda^{*}$, denoted hereon by $\omega_\lambda:=i_\lambda^{*}\omega=\omega\circ i_\lambda$.
\begin{equation}
\xymatrixcolsep{4pc}\xymatrix{
    \Lambda_k(\widetilde{\mathbb{R}^{n}}) \ar[d]_-{\widetilde{\eta}} \ar@{^{(}->}[r]^-{i_\lambda} \ar@/^2pc/[rr]^-{\omega\circ i_\lambda = \omega_\lambda} &\Lambda_{k}(\mathbb{R}\times\widetilde{\mathbb{R}^{n}}) 
    \ar[r]^-{\omega} &\mathbb{R}\\
    \Lambda^{k}(\widetilde{\mathbb{R}^{n}})\ar@{^{(}->}[r]^-{i_\lambda^{*}}&\Lambda^{k}(\mathbb{R}\times\widetilde{\mathbb{R}^{n}})
}
\end{equation} The most important point to be stressed about the pull-back induction (see the schematic diagram) is $\tilde{\eta}$ ensuring the isomorphism between $\Lambda_k(\widetilde{\mathbb{R}^{n}})$ and $\Lambda^{k}(\widetilde{\mathbb{R}^{n}})$. There exists a canonical isomorphism if and only if $\tilde{\eta}$ is non-degenerate (see discussion around Eq. (\ref{o2})). 

\begin{defi}
    Let $\widetilde{A}\subset\widetilde{\mathbb{R}^{n}}$. The so-called Homotopy operator $H$ is a linear application $\Lambda^{k+1}(\mathbb{R}\times \widetilde{{A}^{n}})\to\Lambda^{k}(\widetilde{{A}^{n}})$ such that:
    \begin{enumerate}
        \item If $\omega=\bar{\omega}(\lambda,x^{i_1},...,x^{i_k})_{i_1...i_{k+1}}dx^{i_1}\wedge...\wedge dx^{k+1}$, then
        \begin{equation}H\omega=0;\end{equation}
        \item If $\omega=\bar{\omega}(\lambda,x^{i_1},...,x^{i_k})_{\lambda i_1...i_{k}}d\lambda\wedge dx^{i_1}\wedge...\wedge dx^{k}$, then
        \begin{equation}H\omega=\bigg(\int\limits^{1}_{0}\bar{\omega}(\lambda,x^{i_1},...,x^{i_k})_{\lambda i_1...i_{k}}d\lambda\bigg)dx^{i_1}\wedge...\wedge dx^{k}.
        \end{equation}
    \end{enumerate}
\end{defi} 
Usually, the definition of the domain of $H$ takes it as a star-shaped set. This is unnecessary here, but we still need $\lambda \in \mathbb{R}$; that is, the injection inserts a real parameter into the form coefficient argument. With these definitions and considerations, we can prove the validity of a well-known standard result about the Homotopy operator in our context.   
\begin{lemma}\label{homotopia}
    The Homotopy operator $H$ $(i)$ commutes with $\lambda$ independent $C^\infty$ functions and $(ii)$ $\forall \omega\in\Lambda^{k+1}(\mathbb{R}\times \widetilde{A})$ it satisfies   
           \begin{equation}
            H\dd\omega + \dd H\omega
            =
            \omega_1 - \omega_0,
        \end{equation}
       where $\omega_1=i_1^{*}\omega$ e $\omega_0=i_0^{*}\omega$.
\end{lemma}
\begin{proof} The proof of $(i)$ is trivial. For the second part, the proof will be divided into two parts, where $\omega$ has no dependency on $d\lambda$, $(ii_1)$, and the case $(ii_2)$ where it has a dependency on $d\lambda$.
      
      $(ii_{1})$ For this case $\omega=\bar{\omega}(\lambda,x^{i_1},...,x^{i_k})_{i_1...i_{k+1} }dx^{i_1}\wedge...\wedge dx^{k+1}$, and therefore $H\omega=0$ $\Rightarrow$ $\dd H\omega=0$.
    On the other hand, $H\dd\omega$ will be given by:
    \begin{multline}
        H\dd\omega
        =
        H
        \bigg\{
        \partial_\lambda \bar{\omega}(t,x^{i_1},...,x^{i_k})_{i_1...i_{k+1}}d\lambda\wedge dx^{i_1}\wedge...\wedge dx^{k+1}
        +\\+
        \partial_j \bar{\omega}(\lambda,x^{i_1},...,x^{i_k})_{i_1...i_{k+1}}\dd x^{j}\wedge dx^{i_1}\wedge...\wedge dx^{k+1}
        \bigg\}, 
    \end{multline} so that 
    \begin{equation}
        H\dd \omega
        =
        H
        \bigg(\int\limits^{1}_{0}\partial_\lambda \bar{\omega}(\lambda,x^{i_1},...,x^{i_k})_{i_1...i_{k+1}}dt\bigg)dx^{i_1}\wedge...\wedge dx^{k+1}=\omega_1 - \omega_0
    \end{equation} and therefore $H\dd \omega + \dd H\omega = \omega_{1}-\omega_0$. 
    
    $(ii_{2})$ For this case the acting if $H$ upon $\dd \omega$ is given by 
    \begin{align}
        H\dd \omega 
                &=
        -\sum\limits_{j}
        \bigg(\int\limits^{1}_{0}
        \partial_j
        \bar{\omega}(\lambda,x^{i_1},...,x^{i_k})_{\lambda i_1...i_{k}}d\lambda\bigg)\dd x^{j}\wedge dx^{i_1}\wedge...\wedge dx^{k}.
    \end{align}
    On the other hand
    \begin{align}
    \dd H\omega
        &=
    \bigg(
    \int\limits_{0}^{1}\partial_{j}\bar{\omega}(\lambda,x^{i_1},...,x^{i_k})_{\lambda i_1...i_{k}}d\lambda
    \bigg)\dd x^{j}\wedge dx^{i_1}\wedge...\wedge dx^{i_k},
    \end{align} leading directly to $H\dd\omega+\dd  H\omega=0$. Nevertheless, for a fixed $\lambda$ the pull-back is always null and $\omega_1=\omega_0=0$.
\end{proof}
The presented results pave the way for the study of gauge forms. They may be a starting point for investigating gauge theories (electrodynamics in particular) in nontrivial topologies. While its systematic study shall be addressed someplace else, we intend to finalize this discussion by pointing out that if the gauge connection is denoted by $\tilde{A}=A_\mu dx^\mu+A_\alpha x^\alpha \partial_\mu\theta dx^\mu$, the field strength counterpart, $\tilde{F}$, reads 
\begin{eqnarray}     
\tilde{F}=\frac{1}{2}F_{\mu\nu}dx^\mu \wedge dx^\nu+\big[(A_\mu+x^\alpha\partial_\mu A_\alpha)\partial_\nu\theta+\partial_\alpha A_\nu x^\alpha\partial_\mu\theta\big]dx^\mu \wedge dx^\nu. 
\end{eqnarray} where $F_{\mu\nu}=\partial_\mu A_\nu-\partial_\nu A_\mu$. Consequently, gauge invariance is also lost due to nontrivial topological effects.

\end{document}